\documentclass[conference]{IEEEtran}
\IEEEoverridecommandlockouts
\usepackage[T1]{fontenc}
\usepackage{url}
\usepackage{cite}
\usepackage{graphicx}
\usepackage{amssymb}
\usepackage{eucal}
\usepackage{color}
\usepackage{comment}
\usepackage{arydshln}
\usepackage{algorithm}
\usepackage{algpseudocode}
\usepackage{xcolor}
\algnewcommand\algorithmicinput{\textbf{Input:}}
\algnewcommand\INPUT{\item[\algorithmicinput]}
\algnewcommand\algorithmicoutput{\textbf{Output:}}
\algnewcommand\OUTPUT{\item[\algorithmicoutput]}
\usepackage{mathtools}

\usepackage{enumerate}

\usepackage{amsthm}
\theoremstyle{definition}
\newtheorem{definition}{Definition}
\newtheorem{proposition}{Proposition}
\newtheorem{lemma}{Lemma}
\newtheorem{theorem}{Theorem}
\newtheorem{corollary}{Corollary}
\newtheorem{remark}{Remark}

\newcommand{\cC}{\mathcal{C}}

\newcommand{\remove}[1]{}
\usepackage{tikz}
\usetikzlibrary{positioning}

\newcommand{\ceilenv}[1]{\left\lceil #1 \right\rceil}
\newcommand{\floorenv}[1]{\left\lfloor #1 \right\rfloor}

\newcommand\nc\newcommand
\nc{\vzero}{{\boldsymbol{0}}}
\nc{\bfc}{{\boldsymbol c}}
\nc{\bfr}{{\boldsymbol r}}
\nc\bfa{{\boldsymbol a}}\nc\bfA{{\boldsymbol A}}\nc\cA{{\mathcal A}}
\nc\bfb{{\boldsymbol b}}\nc\bfB{{\boldsymbol B}}\nc\cB{{\mathcal B}}
\nc\bfd{{\boldsymbol d}}\nc\bfD{{\boldsymbol D}}\nc\cD{{\mathcal D}}
\nc\bfe{{\boldsymbol e}}\nc\bfE{{\boldsymbol E}}\nc\cE{{\mathcal E}}
\nc\bff{{\boldsymbol f}}\nc\bfF{{\boldsymbol F}}\nc\cF{{\mathcal F}}
\nc\bfg{{\boldsymbol g}}\nc\bfG{{\boldsymbol G}}\nc\cG{{\mathcal G}}
\nc\bfh{{\boldsymbol h}}\nc\bfH{{\boldsymbol H}}\nc\cH{{\mathcal H}}
\nc\bfi{{\boldsymbol i}}\nc\bfI{{\boldsymbol I}}\nc\cI{{\mathcal I}}
\nc\bfj{{\boldsymbol j}}\nc\bfJ{{\boldsymbol J}}\nc\cJ{{\mathcal J}}
\nc\bfk{{\boldsymbol k}}\nc\bfK{{\boldsymbol K}}\nc\cK{{\mathcal K}}
\nc\bfl{{\boldsymbol l}}\nc\bfL{{\boldsymbol L}}\nc\cL{{\mathcal L}}
\nc\bfm{{\boldsymbol m}}\nc\bfM{{\boldsymbol M}}\nc\cM{{\mathcal M}}
\nc\bfn{{\boldsymbol n}}\nc\bfN{{\boldsymbol N}}\nc\cN{{\mathcal N}}
\nc\bfo{{\boldsymbol o}}\nc\bfO{{\boldsymbol O}}\nc\cO{{\mathcal O}}
\nc\bfp{{\boldsymbol p}}\nc\bfP{{\boldsymbol P}}\nc\cP{{\mathcal P}}
\nc\bfq{{\boldsymbol q}}\nc\bfQ{{\boldsymbol Q}}\nc\cQ{{\mathcal Q}}
\nc\bfs{{\boldsymbol s}}\nc\bfS{{\boldsymbol S}}\nc\cS{{\mathcal S}}
\nc\bft{{\boldsymbol t}}\nc\bfT{{\boldsymbol T}}\nc\cT{{\mathcal T}}
\nc\bfu{{\boldsymbol u}}\nc\bfU{{\boldsymbol U}}\nc\cU{{\mathcal U}}
\nc\bfv{{\boldsymbol v}}\nc\bfV{{\boldsymbol V}}\nc\cV{{\mathcal V}}
\nc\bfw{{\boldsymbol w}}\nc\bfW{{\boldsymbol W}}\nc\cW{{\mathcal W}}
\nc\bfx{{\boldsymbol x}}\nc\bfX{{\boldsymbol X}}\nc\cX{{\mathcal X}}
\nc\bfy{{\boldsymbol y}}\nc\bfY{{\boldsymbol Y}}\nc\cY{{\mathcal Y}}
\nc\bfz{{\boldsymbol z}}\nc\bfZ{{\boldsymbol Z}}\nc\cZ{{\mathcal Z}}

\nc{\entropy}{{\sf H}}
\nc{\maxcode}{{\sf M}}
\nc{\maxrate}{\mu}

\newcommand{\hm}[1]{{\color{red}(HM: #1)}}

\newcommand{\RR}{\mathbb{R}}

\nc{\dist}{{\rm d}}
\nc{\orc}{{\sf OR}}
\nc{\xor}{{\sf XOR}}
\nc{\add}{{\sf ADD}}
\nc{\bflamb}{\boldsymbol{\lambda}}
\mathchardef\mhyphen="2D 
\newcommand{\lambadd}{{\lambda\mhyphen\add}}
\newcommand{\norm}[1]{\left\lVert#1\right\rVert}

\usepackage[top=0.65in, bottom=0.65in, left=0.54in, right=0.54in]{geometry}

\usepackage{hyperref}
\newcommand\myshade{70} 
\hypersetup{ 
	linkcolor  = red!\myshade!black,
	citecolor  = blue!\myshade!black,
	urlcolor   = blue!\myshade!black,
	colorlinks = true,
}

\title{Noise-Tolerant Codebooks for Semi-Quantitative Group Testing:  Application to Spatial Genomics}
\author{
\IEEEauthorblockN{
  Kok Hao Chen\IEEEauthorrefmark{1},
  Duc Tu Dao\IEEEauthorrefmark{1}\IEEEauthorrefmark{2}, 
  Han Mao Kiah\IEEEauthorrefmark{2},
  Van Long Phuoc Pham\IEEEauthorrefmark{2}, 
  Eitan Yaakobi\IEEEauthorrefmark{2}\IEEEauthorrefmark{3}}
\IEEEauthorblockA{\small 
 \IEEEauthorrefmark{1}Genome Institute of Singapore, Agency for Science, Technology and Research (A*STAR)
    }
 \IEEEauthorblockA{\small 
 \IEEEauthorrefmark{2}School of Physical and Mathematical Sciences, Nanyang Technological University, Singapore
    }
 \IEEEauthorblockA{\small \IEEEauthorrefmark{3}Department of Computer Science, Technion --- Israel Institute of Technology, Israel
    }    
  {\footnotesize chenkh@gis.a-star.edu.sg,~daodt@gis.a-star.edu.sg,~hmkiah@ntu.edu.sg,~phuoc.phamvanlong@gmail.com,~yaakobi@cs.technion.ac.il}
  }
\date{}

\begin{document}

\maketitle

\hspace*{-3mm}\begin{abstract}
Motivated by applications in spatial genomics, we revisit group testing (Dorfman~1943)
and propose the class of $\lambda$-{\sf ADD}-codes, studying such codes with certain distance $d$ and codelength $n$.
When $d$ is constant, we provide explicit code constructions with rates close to $1/2$.
When $d$ is proportional to $n$, we provide a GV-type lower bound whose rates are efficiently computable.
Upper bounds for such codes are also studied.
\end{abstract}

\section{Introduction}

Consider $M$ items, of which at most $s$ of them are {\em defectives}.
The objective of group testing is to devise a set of $n$ {\em tests} that can identify this subset of at most $s$ defectives. 
These tests are represented by a binary $(n\times M)$-{\em test matrix} $\bfM$, 
whose rows
and columns are indexed by the $n$ tests and $M$ items, respectively.
Then in the $i$-th test, we look at the $i$-th row of $\bfM$ and 
include in the test the subset $S$ of items whose corresponding entry is one.

A measurement can then be conceptualized as the application of a function $f$ to the subset $S$. 
In the original group testing context, as proposed by Dorfman~\cite{dorfman1943detection}, $f$ takes the value of one if there is at least one defective in $S$. In other words, $f(S)$ represents the Boolean sum of the entries in $S$. 
Later, Csisz{\'a}r and K{\"o}rner introduced the {\em memoryless multiple-access channel} or {\em adder channel}, where $f(S)$ counts the exact number of defectives in $S$ \cite{csiszar2011information}. 
That is, $f(S)$ is simply the integer-valued sum of entries in $S$. 
This model, also known as {\em quantitative group testing}, traces its origins back to the coin-weighing problem \cite{shapiro1960e1399, erdos1963two}.
These models have been widely applied in both theory and practice, and comprehensive surveys can be found in the instructive works of Du and Hwang~\cite{du2000combinatorial} and D'yachkov~\cite{dyachkov2014lectures} (see also, Tables~1,~2 and~3 in~\cite{guruswami2023noise}). 
A recent text with a detailed summary of group testing applications can be found here~\cite{aldridge2019group}.

In this paper, we treat each {\em column} of the test-matrix $\bfM$ as a codeword and refer to the collection of these columns as a {\em code}.
Specifically, with respect to the Boolean sum and the real-valued sum, we refer to these codes as $\orc$-codes and $\add$-codes, respectively (see formal definitions in Section~\ref{sec:prob}).

Later, motivated by applications in genotyping and biosensing, Emad and Milenkovic proposed a novel framework called {\em semi-quantitative group testing (SQGT)}~\cite{emad2014sqgt}.
In SQGT, the measurement $f(S)$ takes on real values, depending on both the number of defectives in $S$ and a specified set of thresholds (a quantizer).
Additionally, the framework allows for the measurement of varying amounts of items, resulting in a non-binary test matrix.
In the same paper, the authors then presented several test matrices or code constructions that correctly perform SQGT in the presence of errors.

{
In this paper, we continue this line of investigation, drawing inspiration from yet another biology application -- spatial genomics. Here, we broadly explain the differences.
In previous biosensing applications~\cite{Sheikh2007,shental2010identification}, the process of inferring the genetic information of a target organism typically involved a fixed microarray of DNA probes. 
The target DNA sample would be fluorescently tagged and then flushed over the microarray. Through the fluorescence values of the microarray, we identify regions where binding occurred, thus inferring the genetic makeup of the target organism. 
However, the procedure loses spatial information of the genetic material on the organism.
In contrast, spatial genomics involves a reversal of roles between the DNA probes and the target genetic material~\cite{goh2020highly}.
In this case, the probes are fluorescently tagged and flushed over the target tissue.
Again, heightened fluorescence values indicate binding, but in this case, we can not only infer the genetic information but also pinpoint the spatial positions of these genetic elements on the target tissue. 

Suppose that there are $M$ possible gene sequences of interest.
One can naively synthesize $M$ different probes and use $M$ testing rounds (of flushing and scanning) to identify each individual sequence.
However, we can adapt an approach {\em \'a la} group testing, or specifically SQGT, and design probes to measure multiple sequences at one time. 
This then significantly reduces the number of testing rounds. Unlike SQGT, our context does not afford the flexibility to vary the amount of tissue.
Consequently, we limit our attention to the design of {\em binary} codewords. 
}

Now, this connection to the SQGT framework is only superficial, and we outline some crucial differences here.
First, we measure errors differently. 
Suppose that $\bff\in\RR^n$ represents the correct measurements from the $n$ tests, and let $\widetilde{\bff}$ be an erroneous measurement. In~\cite{emad2014sqgt}, errors are quantified by the {\em number} of coordinates where $\bff$ and $\widetilde{\bff}$ differ. 
In contrast, our approach considers the $\ell_1$-norm of $\bff-\widetilde{\bff}$, and hence, our objective is to correctly identify the defectives whenever this magnitude is small.

Next, we observe that in~\cite{emad2014sqgt}, all explicit code constructions designed to correct $e>0$ errors utilize $\orc$-codes as their building components, inevitably resulting in a rate loss. 
In contrast, our approach employs a different class of binary codes known as $\xor$-codes (refer to Section~\ref{sec:xor-def}), enabling us to achieve higher rates. For example, when $s=2$ and in the regime where the distance is a constant with respect to the codelength, our approach attains rates of $0.5$, 
surpassing best known construction of $\orc$-codes with rates $0.302$ (see also, Figure~\ref{fig:GVplot}).

In the regime where the number of errors grows linearly with the codelength, the rates achievable by the analysis in~\cite{emad2014sqgt} are unclear. 
Here, we perform a different analysis, using the probabilistic method on hypergraphs and analytic combinatorics in several variables (ACSV), to derive computable lower bounds for achievable rates. 

In summary, we revisit the SQGT framework in the context of the $\ell_1$-distance and 
terming these codes {\em $\lambadd$-codes}. 
In this work, our focus lies in constructing $\lambadd$-codes with distance $d>1$.
First, when $d$ is a constant with respect to the  codelength, using $\xor$-codes, we construct explicit families of $\lambadd$-codes that achieve higher rates than any known construction. 
Second, when $d$ grows linearly with the codelength, we provide a GV-type lower bound where the achievable rates can be efficiently computed.


\section{Problem Statement}
\label{sec:prob}

For a positive integer $n$, we use $[n]$ and $\{0,1\}^n$ to denote the set $\{1,2,\ldots, n\}$ and the set of length-$n$ binary strings, respectively.
For some finite set $\cX$ and integer $s\le |\cX|$, we use $\binom{\cX}{\le s}$ to denote the collection of all {\em nonempty} subsets of $\cX$ with size at most $s$.

As mentioned earlier, the integers $n$, $M$, and $s$ denote the number of tests, items, and the maximum number of defectives, respectively. Furthermore, we introduce a function $\bff$ whose domain is $\binom{\{0,1\}^n}{\le s}$ and codomain is $\RR^n$. The function $\bff$ depends on the size of the input subset and reflects the measurement process in the experiment. Given a code $\cC\subseteq\{0,1\}^n$, its {\em $\bff$-norm} is defined to be $\dist_\bff(\cC;s)\triangleq\min\{\norm{\bff(S_1)-\bff(S_2)} : S_1,S_2\in \binom{\cC}{\le s},~S_1\ne S_2 \}$, 
where $\norm{\cdot}$ is the $\ell_1$-distance. In this paper, we fix $n$, $d$, $s$, and $\bff$, and our task is to find a code $\cC\subseteq\{0,1\}^n$, such that $\dist_\bff(\cC;s)$ is at least $d$. Then classical group testing requires one to find a code $\cC$ with $\dist_\bff(\cC;s)>0$.


As always, we are interested in maximizing the code size. That is, we are interested in determining
\begin{equation}\label{eq:Afnd}
\maxcode_\bff(n,d;s) \triangleq \max\{|\cC|\,:\, \cC\subseteq\{0,1\}^n,~\dist_\bff(\cC;s)\ge d\}\,.
\end{equation}
Two regimes of interest are as follows. 
\begin{itemize}
\item When $d$ is constant with respect to $n$, we estimate the optimal {\em redundancy} $\rho_\bff(n,d;s)\triangleq n-\log \maxcode_\bff(n,d;s)$. Here and in the rest of the paper, the logarithm is taken base two.
\item For fixed $\delta>0$, we estimate the optimal {\em asymptotic rate} $\maxrate_\bff(\delta;s) \triangleq \limsup_{n\to\infty} \frac1n \log \maxcode_\bff(n,\floorenv{\delta n};s)$.
In the case for $\delta=0$, we estimate 
$\maxrate_\bff(s) \triangleq \limsup_{n\to\infty} \frac1n \log \maxcode_\bff(n,1;s)$.
\end{itemize}

In what follows, we provide a short summary of known codebook constructions in these asymptotic regimes. Notable works on group testing in the presence of noise can be found in~\cite{cheraghchi2009noise, bshouty2012, goshkoder2024efficient}. 

\subsection{$\orc$-Codes and $\xor$-Codes}
\label{sec:xor-def}

In the original group testing setup defined by Dorfman~\cite{dorfman1943detection}, the function $\bff$ corresponds to the Boolean sum. Specifically, we define the $\orc$ function such that $\orc(S) \triangleq \bigvee_{\bfx\in S} \bfx$ for $S\in\binom{\{0,1\}^n}{\le s}$. 
The corresponding rates $\maxrate_\orc(s)$ have been extensively studied.
Here, we state its best-known estimates and we refer the interested reader to~\cite{du2000combinatorial, dyachkov2014lectures} for more details.

\begin{theorem}[{see \cite[Chapter 7]{du2000combinatorial} or \cite[Section 2]{dyachkov2014lectures}}]
\label{thm:or}
There are computable constants $L_s$ and $U_s$ 
such that $L_s\le \maxrate_{\orc}(s)\le U_s$ for $s\ge 2$. Here, we list the first few values of $L_s$ and $U_s$. 
\begin{center}
\begin{tabular}{|c||c|c|c|c|c|}
\hline
$s$    & 2 & 3 & 4 & 5 & 6 \\ \hline
$U_s$  & 0.500 & 0.333 & 0.250 & 0.200 & 0.167\\
$L_s$  & 0.302 & 0.142 & 0.082 & 0.053 & 0.037 \\ \hline
\end{tabular}
\end{center}
\end{theorem}

To construct $\orc$-codes with distance strictly greater than one, we can use certain intersection properties of combinatorial packings (see \cite[Cor.~8.3.3]{du2000combinatorial}).
The following proposition is an asymptotic version of a result from Erd\"os, Frankel and F\"uredi~\cite{erdos1985families} (see also~\cite[Thm.~7.3.8]{du2000combinatorial}). 
For completeness, we provide a detailed derivation in~Appendix~\ref{app:lower-or}. 

\begin{proposition}\label{prop:lower-or}
For $\delta\ge 0$, we have that 
$\maxrate_\orc(\delta;s)\ge 
\max_{\kappa\ge \delta} H\left(\frac1s(\kappa-\delta)\right) 
- 2\kappa H\left(\frac1{\kappa s}(\kappa-\delta)\right)$\,.  
\end{proposition}

To construct codes for later sections, we introduce the class of $\xor$-codes.
Specifically, here the function $\bff$ corresponds to the $\xor$ sum (or addition modulo two). 
In other words, we define the $\xor$ function such that $\xor(S) \triangleq \bigoplus_{\bfx\in S} \bfx$ for $S\in\binom{\{0,1\}^n}{\le s}$. To the best of our knowledge, this class of codes has been studied before only for $d=1$ (see e.g.~\cite[Section V]{Sloane_Bh},~\cite{OBryant2004}). 
Unlike $\orc$-codes, we are able to provide asymptotically sharp estimates for the rates of $\xor$ codes. 
Specifically, for $0\le d\le n$, let $A(n,d)$ and $A^*(n,d)$ 
be the maximum size of a binary code and a binary {\em linear} code, respectively, of length $n$ and Hamming distance $d$. 
Furtheremore, we set  $\alpha(\delta)\triangleq \lim\sup_{n\to\infty} \frac1n \log A(n,d)$ and $\alpha^*(\delta)\triangleq \lim\sup_{n\to\infty} \frac1n \log A^*(n,d)$\,.
In Section~\ref{sec:bounds-xor}, we prove the following bounds.

\begin{theorem}\label{thm:main-xor}
Let $0\le d\le n$. 
Then,
{
\begin{equation}\label{eq:main-xor}
\hspace{-0.5ex}\Big(A^*(n,d)\Big)^{1/s}\hspace{-2.25ex}-1 \hspace{-0.2ex}\le\hspace{-0.2ex} \maxcode_\xor(n,d;s) \hspace{-0.2ex}\le\hspace{-0.2ex} \Big(s! A(n,d)\Big)^{1/s}\hspace{-1.25ex}+s-1\,.   
\end{equation}
}
\end{theorem}

Then taking logarithms, we obtain the following sharp estimates.
\begin{corollary}\label{cor:xor} \hfill
\begin{enumerate}[(i)]
\item For every fixed $d$ with respect to $n$, it holds that $\rho_\xor(n,d;s)=\left(1-\frac1s\right)n + \frac1s\floorenv{\frac{d-1}{2}}\log n + o(\log n)$\,.
\item Fix $0\le \delta\le 1$. 
Then $\frac1s \alpha^*(\delta) \le \maxrate_\xor(\delta;s) \le \frac1s \alpha(\delta)$\,.
\end{enumerate}
\end{corollary}

\subsection{$\lambda$-$\add$-Codes}

We define our main object of study -- {\em $\lambadd$-codes}.
To this end, we let $\bflamb=(\lambda_0,\lambda_1,\ldots, \lambda_s)\in\RR^{s+1}$ such
that $\lambda_0=0$ and $\lambda_1=1$. 
We first define the operation $\boxplus_{\lambda}$ on a multi-set of $s$ bits as follows
\begin{equation}
x_1 \boxplus_{\bflamb} x_2 \boxplus_{\bflamb} \cdots \boxplus_{\bflamb} x_s = \lambda_{t}, \text{ where } t = |\{ i :~x_i=1\}|\,.
\end{equation}
Then, for $s'\leq s$, the operation $\boxplus_{\bflamb}$ on the set $\{\bfx_1,\ldots, \bfx_{s'}\}\in \binom{\{0,1\}^n}{\le s}$ 
is defined component-wise. So, here, the function $\bff$ of $\lambadd$-codes corresponds to $\boxplus_{\bflamb}$.

Now, we note that we recover previous code classes by specializing the $(s+1)$-tuples $\bflamb$.
\begin{itemize}
\item When $\bflamb=(0,1,1,\ldots,1)$, we obtain $\orc$-codes.
\item When $\bflamb=(0,1,0,\ldots,s\bmod{2})$, we obtain $\xor$-codes.
\item When $\bflamb=(0,1,2,\ldots,s)$, we obtain $\add$-codes defined in~\cite{csiszar2011information}.
\end{itemize}
 When $\bflamb=(0,1,2,\ldots,s)$, we simply refer to $\lambadd$-codes as $\add$-codes and the corresponding rates $\maxrate_\add(s)$ have been extensively studied. 
 As before, we state its best-known estimates and we refer the interested reader to~\cite{dyachkov2014lectures} and the references therein for further details.

\begin{theorem}[{see \cite[Section 5]{dyachkov2014lectures}}]
\label{thm:or}
There are computable constants $L^\add_s$ and $U^\add_s$ 
such that $L^\add_s\le \maxrate_{\add}(s)\le U^\add_s$ for $s\ge 2$. Here, we list the first few values of $L^\add_s$ and $U^\add_s$. 
\begin{center}
\begin{tabular}{|c||c|c|c|c|c|}
\hline
$s$    & 2 & 3 & 4 & 5 & 6 \\ \hline
$U^\add_s$  & 0.600 & 0.562 & 0.450 & 0.419 & 0.364\\
$L^\add_s$  & 0.500 & 0.336 & 0.267 & 0.225 & 0.195 \\ \hline
\end{tabular}
\end{center}
\end{theorem}

We remark that for $s=2$, the lower bound 0.5 is obtained by Lindstr{\"o}m using results from additive number theory~\cite{lindstrom1969determination}.
For $s\ge 3$, the lower bound is obtained using the random coding method~\cite{dyachkov1981coding,poltyrev1987improved}.
However, for $\add$-codes with distance strictly greater than one, we are unaware of any results.
Nevertheless, for $s=2$, we obtain a GV-type lower bound using an analysis similar to the random coding method.

Specifically, in Section~\ref{sec:bounds-lamb}, we derive a GV-type lower bound for arbitrary general $\bflamb=(1,\lambda)$.
In the same section, we also obtain upper and lower bounds for the size of $\lambda$-$\add$-codes. Of significance, in our constructions, we make use of $\xor$-codes and for certain regimes, the rates of such codes exceed those guaranteed by random coding.
Such comparisons are given in Section~\ref{sec:comparison}.

\section{Bounds For $\xor$-Codes}
\label{sec:bounds-xor}
In this section, we examine $\xor$ codes and in particular, provide complete proofs for Theorem~\ref{thm:main-xor}.
First, we demonstrate the lower bound of \eqref{eq:main-xor} and to this end, we have the following construction.

\begin{theorem}\label{thm:constr-xor}
If there exists an $[M,M-k,2s+1]$-linear code  $\cC_k$ and an $[n,k,d]$-linear code $\cC_n$,
then there exists a code $\cC$ of length $n$ and size $M$ such that $\dist_\xor(\cC)$ is at least $d$.
\end{theorem}

\begin{proof}
Let $\bfh_1\,\ldots,\bfh_M$ be the length-$k$ column vectors of a parity check matrix $\bfH_k$ of $\cC_k$.
Also, let $\bfG_n$ be a $k\times n$-generator matrix of $\cC_n$. Then we set
\[\cC \triangleq \{ \bfG^T\bfh_i  : i\in [M]\}\,.\]

We claim that $\dist_\xor(\cC)$ is at least $d$. 
Now, since $\cC$ is a subcode of $\cC_n$, the $\xor$-distance of $\cC$ is at least $d$,
provided that $\xor(S)-\xor(S')$ is nonzero for any two distinct subsets $S,S'\in \binom{\cC}{\le s}$.
Suppose to the contrary that $\xor(S)-\xor(S')=\vzero$.
Since $\bfG$ is a generator matrix and each summand in $S$ or $S'$ is of the form $\bfG^T\bfh_i$, we have $s'$ columns in $\bfH_k$ summing to $\vzero$ with $s'\le 2s$.
This then contradicts the fact that $\bfH_k$ is a parity check matrix for the $[M,M-k,2s+1]$-code $\cC_k$.
\end{proof}

To complete the proof of the lower bound of Theorem~\ref{thm:main-xor}, we set $\cC_k$ to be a shortened narrow-sense BCH code (see~\cite[Ch. 9]{MS1977}). 
Specifically, for $n$, $d$, $s$, we first find an $[n,k,d]$-linear code with dimension at least $k\triangleq\log A^*(n,d)$. Then, we set $M$ such that $s\log(M+1)\ge k$ and we have a shortened $[M,M-s\log(M+1), \ge 2s+1]$-BCH code. 
Applying Theorem~\ref{thm:constr-xor}, we obtain the lower bound in~\eqref{eq:main-xor}.

Next, we derive the upper bound in~\eqref{eq:main-xor}.
\begin{proof}[Proof of Upper Bound in Theorem~\ref{thm:main-xor}]
For brevity, we denote the value $\maxcode_\xor(n,d;s)$ by $M$, and suppose that $\cC$ is a code with $d_\xor(\cC)\ge d$ and $|\cC|=M$.
We first prove the following:
\begin{equation}\label{eq:upper-xor}
\sum_{s'=1}^s \binom{M}{s'}\le A(n,d)\,.    
\end{equation}
Let $\cC_\xor = \{ \xor(S)~:~S\in \binom{\cC}{\le s}\}$.
Since the image $\xor(S)$ remains in $\{0,1\}^n$, we have that $\dist_\xor(\cC)$ is the Hamming distance of $\cC_\xor$. In other words, the size of the code $\cC_\xor$ is at most $A(n,d)$, as required. Lastly, since $\sum_{s'=1}^s \binom{M}{s'}\ge \binom{M}{s}\ge \frac{(M-s+1)^s}{s!}$, the inequality in~\eqref{eq:main-xor} follows.
\end{proof}

\section{Bounds for $\lambda$-$\add$-Codes}
\label{sec:bounds-lamb}

\begin{table}[!t]
\centering
\begin{tabular}{|c|cccc|}\hline
    $u$ & 0 & 0 & 1 & 1  \\ 
    $v$ & 0 & 1 & 0 & 1  \\ \hline 
    $u \oplus v$ &  0 & 1 & 1 & 0 \\ 
    $u \vee v$ &  0 & 1 & 1 & 1 \\
    $u + v$ &  0 & 1 & 1 & 2 \\
    $u \boxplus_{\bflamb} v$ &  0 & 1 & 1 & $\lambda$ \\ \hline
\end{tabular}
\vspace{1mm}

\caption{List of binary operations considered in this paper.}
\label{table:operations}
\end{table}

In this section, we study the class of $\lambda$-$\add$ codes for the case $s=2$.
We defer the investigation for $s>2$ to future work.
Henceforth, we set $\bflamb=(0,1,\lambda)$ with $1<\lambda\le 2$ and estimate the quantity $\maxcode_{\lambadd}(n,d;2)$ for $d$ and $n$.

First, we consider a code $\cC$ and study its $\lambadd$-distance relative to its $\orc$-, $\xor$-, and $\add$-distances.

\begin{proposition}\label{prop:lamb-others}
    Set $\bflamb=(0,1,\lambda)$ with $1<\lambda\le 2$.
    If $\cC$ is a binary code, then we have that
    \begin{equation}\label{eq:lambda-xor}
        \dist_{\lambadd}(\cC)\ge (\lambda-1)\dist_{\add}(\cC)\ge (\lambda-1)\dist_{\xor}(\cC)\,. 
    \end{equation}
    Also, we have that 
    \begin{equation}
        \dist_{\lambadd}(\cC)\ge \dist_{\orc}(\cC)\,.
    \end{equation}
    Therefore, $\maxcode_\lambadd (n,d;2) \ge \maxcode_\xor\left(n,\ceilenv{\frac{d}{\lambda-1}};2\right)$  and 
    $\maxcode_\lambadd (n,d;2) \ge \maxcode_\orc\left(n,d;2\right)$.
\end{proposition}

\begin{proof}
Since all distances are computed componentwise, it suffices to check that the inequalities hold for each component. In other words, for $u,v,a,b\in\{0,1\}$, we need to check that
\begin{align*}
    \norm{(u\boxplus_\lambda v)-(a\boxplus_\lambda b)} &\geq (\lambda-1) \norm{(u+v)-(a+b)},\\
    \norm{(u + v)-(a + b)} &\geq  \norm{(u \oplus v)-(a\oplus b)},\\
    \norm{(u\boxplus_\lambda v)-(a\boxplus_\lambda b)} &\geq  \norm{(u\vee v)-(a\vee b)}\,.
\end{align*}
We then verify these inequalities using Table~\ref{table:operations}.
\end{proof}
Therefore, using $\orc$-, $\xor$-, and $\add$-codes, we are able to construct $\lambadd$-codes.
Next, we derive a GV-type lower bound for $\lambadd$-codes in the regime where $d$ is proportional to $n$, 
and in Section~\ref{sec:comparison}, we compare all four constructions.

\subsection{A GV-Type Lower Bound}
\label{sec:lamb-gv}

We provide GV-type lower bounds in the asymptotic regime $d=\delta n$ where $\delta$ is a fixed constant.
When $\delta=0$, similar lower bounds can be obtained via {\em random coding} arguments found in \cite{dyachkov1981coding} (see also the notes in \cite[Section~5.3]{dyachkov2014lectures}). 
Here, we provide an argument using the delightful {\em probabilistic method} (see \cite[Section~3.2]{alon2016probabilistic}). 
Specifically, we use the method to provide a lower bound on the independence number of certain hypergraphs.

Formally, a {\em hypergraph $\cG = (\cV,\cE)$} is defined by a set $\cV$ of {\em vertices}, and a collection $\cE$ of subsets of $V$, known as {\em hyperedges}.
A subset $S\subseteq \cV$ of vertices is {\em independent} if no $i$ vertices in $S$ form a hyperedge of size $i$ in $\cE$ (for $1\le i\le |\cV|$). 
The {\em independence number $\sigma(\cG)$} 
is given by the size of a largest independent set.

For our coding problem, we consider a hypergraph $\cG(n,d) = (\cV(n,d),\cE(n,d))$. Here, $\cV(n,d)$ is comprised of all binary words of length $n$ and $\cE(n,d)$ comprises certain $2$-, $3$-, and $4$-subsets of $\cV(n,d)$. Specifically,
\begin{enumerate}[(G1)]
\item $\{\bfa,\bfb\}$ belongs to $\cE(n,d)$ if and only if $\norm{\bfa-\bfb} < d$.
\item $\{\bfa,\bfb,\bfc\}$ belongs to $\cE(n,d)$ if and only if 
$\norm{(\bfa\boxplus_{\lambda}\bfb)-\bfc} < d$, 
$\norm{(\bfa\boxplus_{\lambda}\bfc)-\bfb} < d$, or,
$\norm{(\bfb\boxplus_{\lambda}\bfc)-\bfa} < d$.

\item $\{\bfa,\bfb,\bfc,\bfd\}$ belongs to $\cE(n,d)$ if and only if 
$\norm{(\bfa\boxplus_{\lambda}\bfb)-(\bfc\boxplus_{\lambda}\bfd)} < d$, 
$\norm{(\bfa\boxplus_{\lambda}\bfc)-(\bfb\boxplus_{\lambda}\bfd)} < d$, or
$\norm{(\bfa\boxplus_{\lambda}\bfd)-(\bfb\boxplus_{\lambda}\bfc)} < d$. 
\end{enumerate}

Given the hypergraph $\cG(n,d)$, we see that an independent set $\cC\subseteq \cV(n,d)$ yields a code $\cC$ with $\dist_\lambadd(\cC;2)\ge d$.
Therefore, our task is to obtain a lower bound for $\sigma(\cG(n,d))$.
To this end, we demonstrate the following lower bound for general hypergraphs.
The proof adapts the lower bound for usual graphs to hypergraphs.

{
\begin{theorem}\label{thm:hypergraphs}
Let $\cG=(\cV,\cE)$ be a hypergraph whose hyperedges have sizes belonging to $\{2,3,4\}$. Suppose that  $|\cV|=N$ and for $i\in \{2,3,4\}$, the number
of hyperedges of size $i$ is given by $E_i$. 
Then
\begin{equation}\label{eq:GVsize}
\sigma(\cG)\ge \min\left\{\frac1{16} \left(\frac{N^i}{E_i}\right)^{1/(i-1)}: i\in\{2,3,4\}\right\}\,.
\end{equation}
\end{theorem}

Before we provide the proof, we derive a GV-type lower bounds for $\lambadd$ codes.

\begin{corollary}\label{cor:gvrate}.
Fix $\delta$ and for $i\in\{2,3,4\}$, 
suppose that $\beta_i = \lim_{n\to \infty} (\log E_i)/n$. Then
\begin{equation}\label{eq:GV-lambadd}
\maxrate_\lambadd(\delta;2) \ge \min\left\{\frac{1}{i-1}(i-\beta_i): i \in \{2,3,4\}\right\}\,.   
\end{equation}
\end{corollary}
}

{
\begin{proof}[Proof of Theorem~\ref{thm:hypergraphs}]
We randomly and independently pick vertices from $\cV(n,d)$ with probability $p$ and put them into a set $S$. We determine the value of $p$ later.

First, we describe how to construct an independent set from $S$. For each hyperedge of size $i$ in the induced subgraph $\cG|_S$, we remove one vertex. After removing these vertices, it is clear that the remaining vertices form an independent set.

In what follows, we determine a lower bound on the number of surviving vertices.
To this end, we consider the following random variables (r.v.). 
Set the r.v. $X$ be the number of the vertices in $S$.
For $i\in\{2,3,4\}$, we set the r.v. $Y_i$ to be the number of hyperedges of size $i$ in the induced subgraph $\cG|_S$. Then, using linearity of expectations, we have that $E[X] = Np$, while $E[Y_i] = p^i E_i$.
Therefore, the expected number of surviving nodes is at least 
\begin{equation}\label{eq:hypergraph1}
Np - p^2 E_2 - p^3 E_3 - p^4 E_4. 
\end{equation}
Here, we choose $p$ to minimize \eqref{eq:hypergraph1}.
If we choose $p = \tau/(2N)$, then \eqref{eq:hypergraph1} becomes:
\begin{equation}\label{eq:hypergraph2}
   \frac{\tau}{2} - \frac{\tau}{4}\left( \frac{\tau E_2}{N^2}\right) - \frac{\tau}{8} \left(\frac{\tau^2 E_3}{N^3}\right) - \frac{\tau}{16} \left(\frac{\tau^3 E_4}{N^4}\right).
\end{equation}
For all $i\in\{2,3,4\}$, since $\tau\le (N^i/E_i)^{1/(i-1)}$, we have that $\tau^{i-1} E_i\le N^i$.
Thus, \eqref{eq:hypergraph2} is least
\begin{equation*}
    \frac{\tau}{2}\left(1-\frac 12-\frac14 -\frac18\right)=\frac{\tau}{16}.
\end{equation*}
Hence, $p$ can always be chosen so that \eqref{eq:hypergraph1} is at least $\tau/16$.
In other words, there is an independent of size at least $\tau/16$.
\end{proof}

Therefore, to provide a GV-type lower bound, we enumerate $E_i$ and determine $\beta_i=\lim_{n\to\infty} (\log E_i)/n$ for $i\in\{2,3,4\}$.
To this end, we borrow tools from analytic combinatorics in several variables (see the text~\cite{melczer2021invitation}).
As the tools are rather involved and technical, and due to space constraints, we defer the detailed derivations to~Appendix~\ref{app:hyperedge}.

Here, we state the results.

\begin{theorem}\label{thm:lamb-GV}
Let $0\le \delta\le 1/2$ and $\lambda=a/b$ with $a$ and $b$ integers. 
Set $\beta_2 = \entropy(\delta) + 1$, where $\entropy$ is the usual binary entropy function. Set 
\begin{align*}
H_3(z,u)        &= 1 - 3z  -  zu^a - 3zu^b -  zu^{a-b},\\
G_{(3,1)}(z, u) &= 1 - z   -  zu^a  -zu^b  - zu^{a-b},\\
G_{(3,2)}(z, u) &= 1 - 2z           -zu^b  - zu^{a-b},\\
H_4(z,u)        &= 1 - 6z  - 2zu^a - 4zu^b - 4zu^{a-b},\\
G_{(4,1)}(z, u) &= 1 - 2z  - 2zu^a -2zu^b -  2zu^{a-b},\\
G_{(4,2)}(z, u) &= 1 - 4z  - 2zu^a         -  zu^{a-b}.
\end{align*}
For $i\in\{3,4\}$, let $(z_i^*,u_i^*)$ be the unique solution with $(z_i^*,u_i^*)>0$ satisfying
\[ H_i(z,u) = 0 \text{ and } u\frac{\partial}{\partial u}H_i(z,u) =   b\delta z\frac{\partial}{\partial z}H_i(z,u)\, \]
\noindent and set $\beta_i = - \log z_i^* - b\delta \log u_i^*$\,.
For $(i,j)\in\{(3,1), (3,2)$, $(4,1), (4,2)\}$, we similarly define 
$\left(z_{(i,j)}^*,u_{(i,j)}^*\right)$ to be the unique solution with $\left(z_{(i,j)}^*,u_{(i,j)}^*\right)>0$ satisfying
\[ G_{(i,j)}(z,u) = 0 \text{ and } u\frac{\partial}{\partial u}G_{(i,j)}(z,u) =   b\delta z\frac{\partial}{\partial z}G_{(i,j)}(z,u)\, \]
\noindent and set $\gamma_{(i,j)} = - \log z_{(i,j)}^* - b\delta \log u_{(i,j)}^*$\,.

If $\beta_i> \gamma_{(i,j)}$ for all $(i,j)\in\{(3,1), (3,2), (4,1), (4,2)\}$, then \eqref{eq:GV-lambadd} holds.
\end{theorem}

\begin{remark}\hfill
\begin{itemize}
\item When the hyperedges of $\cG$ are all of size two, Theroem~\ref{thm:hypergraphs} yields a weaker version of the generalized GV bound~\cite{tolhuizen1997generalized}. As pointed out in the same paper, a better bound can be obtained by {\em Tur\'an's theorem}~\cite[Thm 3.2.1]{alon2016probabilistic}. Extensions of Tur\'an's theorem for general hypergraphs later derived in~\cite{caro1991improved,thiele1999lower,csaba2012note}. However, the determination of asymptotic rates arising from these improvements ``seems in general a hopeless task'' (see discussion in Section~3 of~\cite{tolhuizen2011generalisation}).
\item When $\delta=0$ and $\lambda=2$, we recover the random coding bound in \cite[Sect.~5.1.2, Thm.~3]{dyachkov2014lectures} for $s=2$.
\end{itemize}
\end{remark}

\subsection{Comparison of Lower Bounds for $\lambadd$-Codes}
\label{sec:comparison}

We compare the lower bounds that we obtained in this paper.
Specifically, for illustrative purposes, we fix $\lambda=4/3$ and plot the following curves in Figure~\ref{fig:GVplot}.
\begin{itemize}
\item First, we set $a=4$ and $b=3$ and construct a $\lambadd$-codes with $\lambda=4/3$. Then the GV-type bound given by Theorem~\ref{thm:lamb-GV} provides a lower bound for $\mu_\lambadd$.
\item Next, we consider a lower bound obtained by constructing an $\add$-code. Here, we set $a=2$ and $b=1$ and we determine the GV-type bound given by Theorem~\ref{thm:lamb-GV}. Then it follows from Proposition~\ref{prop:lamb-others} that we have $\maxrate_\lambadd(\delta;2)\ge \maxrate_\add(3\delta;2)$.
\item Next, we consider a lower bound obtained by constructing an $\xor$-code. In other words, we use Proposition~\ref{prop:lamb-others} that states $\maxrate_\lambadd(\delta;2)\ge \maxrate_\xor(3\delta;2)$. 
From Theorem~\ref{thm:main-xor}, we have that $\maxrate_\xor(3\delta;2)\ge \frac12\alpha^*(3\delta)$. 
Since the Gilbert-Varshamov (GV) bound for linear codes~\cite{varshamov1957estimate} states that $\alpha^*(\delta)\ge 1-\entropy(\delta)$, we have that $\maxrate_\xor(3\delta;2)\ge \frac12 (1-\entropy(3\delta))$ and we plot the corresponding curve in Figure~\ref{fig:GVplot}.
\item Finally, we consider a lower bound obtained by constructing an $\orc$-code. In other words, we use Proposition~\ref{prop:lamb-others} that states $\maxrate_\lambda(\delta;2)\ge \maxrate_\orc(\delta;2)$.
So, here, we plot the curve according to Proposition~\ref{prop:lower-or}.
\end{itemize}

Unsurprisingly, $\orc$-codes exhibit the lowest rates due to their challenging construction.
For most values of $\delta$, we observe that applying the GV analysis directly to $\lambda=4/3$ yields significantly better rates compared to an indirect application to $\add$-codes. 
This serves as justification for introducing the class of $\lambadd$-codes. 
Finally, when $\delta$ is small, the rates obtained from using $\xor$-codes are slightly higher.  Our future work involves investigating whether this observation can be leveraged to improve the probabilistic construction method outlined in Theorem~\ref{thm:hypergraphs}.


\begin{figure}
    \centering
    \includegraphics[width=9cm]{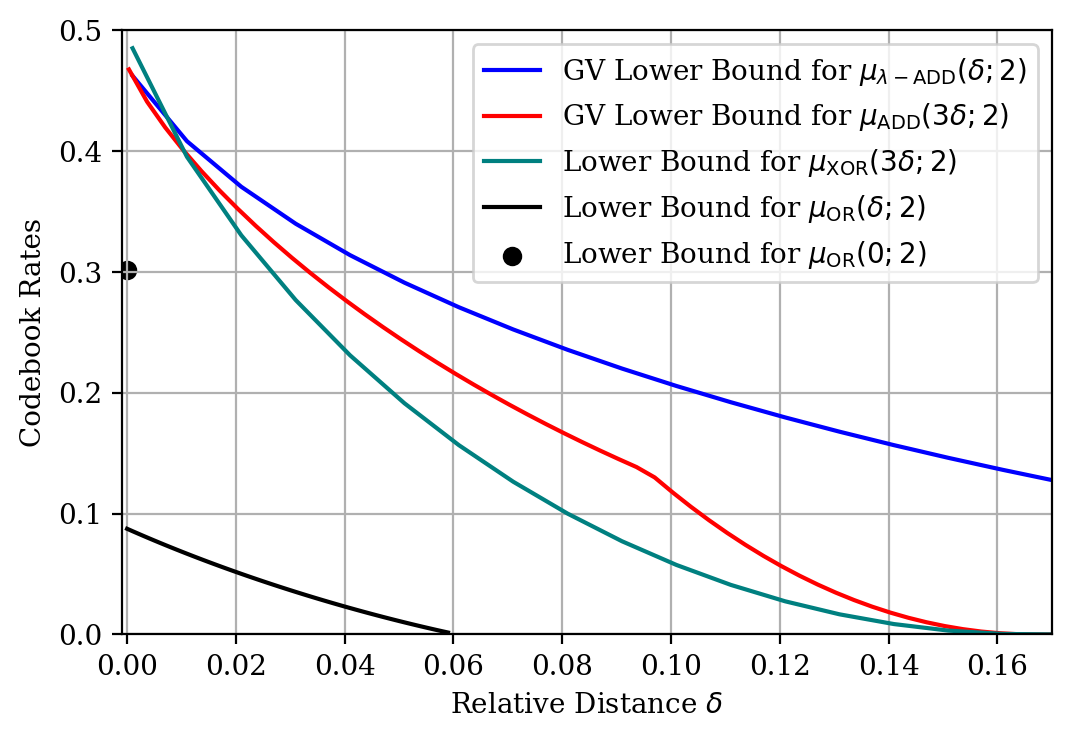}
    \vspace{-5mm}
    
    \caption{Comparison of lower bounds for $\mu_\lambadd(\delta;2)$. Note that from Proposition~\ref{prop:lamb-others}, we have that  $\maxrate_\lambadd(\delta;2)\ge \maxrate_\add(3\delta;2)\ge \maxrate_\xor(3\delta;2)$ and $\maxrate_\lambadd(\delta;2)\ge\maxrate_\orc(\delta;2)$. 
    \vspace{-5mm}
    }
    \label{fig:GVplot}
\end{figure}

\subsection{Upper Bound for $\add$-Codes}

For completeness, we derive upper bounds for the code sizes. 
Here, we keep our exposition simple by setting $\lambda=2$. In other words, we study $\add$-codes.

\begin{proposition}\label{eq:upper-lamb}
Let $0\le d\le n$. Let $A_3(n,d)$ be the maximum size of a ternary code of length $n$ and $\ell_1$-distance $d$. 
Then,
\begin{equation}\label{upper:lamb}
\maxcode_\add(n,d;2) \le  \Big(2 A_3(n,d)\Big)^{1/2}\,.
\end{equation}
\end{proposition}

\begin{proof}
For brevity, we set $\maxcode_\add(n,d)$ to be $M$, and suppose that $\cC$ is a code with $d_\add(\cC)\ge d$ with $|\cC|=M$.
Consider the code $\cC_\add = \{ \add(S)~:~S\in \binom{\cC}{\le 2}\}$.
Now, the image $\add(S)$ belongs $\{0,1,2\}^n$, we have that $\add(\cC)$ a ternary code of length $n$ and $\ell_1$-distance $d$. In other words, the size of $\cC_\add$ is at most $A_3(n,d)$, and so, $M+\binom{M}{2}\le A_3(n,d)$. Then~\eqref{eq:upper-lamb} follows from direct manipulations.
\end{proof}

Next, we consider the regime where $d$ is constant with respect to $n$.
Here, we set $t=\floorenv{(d-1)/2}$.
Since it is easy to show that $A_3(n,d) = O(3^n/n^t)$ 
we have 
\[\rho_\add(n,d) \ge \left(1-\frac{\log 3}{2}\right) n + \frac t2\log n - O(1) \approx 0.208 n \,. \]
On the other hand, from Proposition~2, we have that $\rho_\add(n,d)\le \rho_\xor(n,d)$. Then applying Corollary~\ref{cor:xor}, we have that 
\[\rho_\add(n,d) \le  \frac12 n + \frac t2\log n - o(\log n) \approx 0.5 n\,. \]

Unfortunately, we have a significant gap. 



\section{Conclusion}

Motivated by applications in spatial genomics, we revisit group testing and the SQGT framework in the context of the $\ell_1$-distance.
We propose the class of $\lambadd$-codes and construct $\lambadd$-codes with $d>1$.

In the process of constructing these codes, we study $\xor$-codes and provide sharp estimates for the rates of such codes in terms of classical parameters (see Theorem~\ref{thm:main-xor}). 
Furthermore, employing $\xor$-codes, we construct explicit families of $\lambadd$-codes that achieve higher rates than any known construction (see Corollary~\ref{cor:xor}).
When $d$ grows linearly with the codelength, we provide a GV-type lower bound (see Theorem~\ref{thm:lamb-GV}) that can be efficiently computed.

We vary the relative distance $\delta$ and compare the various constructions of $\lambadd$-codes in Figure~\ref{fig:GVplot}. 
We see that the GV-lower bound generally provides the best estimates, but curiously, the construction for $\xor$-codes performs well when $\delta$ is small.

\section{Acknowledgement}
The work of Han Mao Kiah was supported by the Ministry of Education, Singapore, under its MOE AcRF Tier~2 Award under Grant MOE-T2EP20121-0007 and MOE AcRF Tier~1 Award under Grant RG19/23. The work of Van Long Phuoc Pham was supported by the Ministry of Education, Singapore, under its MOE AcRF Tier~1 Award under Grant RG19/23.
The research of Eitan Yaakobi was Funded by the European Union (ERC, DNAStorage, 865630). Views and opinions expressed are however those of the author(s) only and do not necessarily reflect those of the European Union or the European Research Council Executive Agency. Neither the European Union nor the granting authority can be held responsible for them. 

\newpage

\bibliographystyle{IEEEtran}

\bibliography{references}

\newpage


\appendices

\section{Proof of Proposition~\ref{prop:lower-or}}\label{app:lower-or}

In this appendix, we provide a detailed derivation of Proposition~\ref{prop:lower-or}.
To this end, we introduce the class of combinatorial {\em packing}.

\begin{definition}\label{def:packing}
A {\em $t$-$(v,k,1)$-packing} is a pair $(V,\cB)$, where 
$V$ is a set of $v$ {\em vertices} and $\cB$ is a collection of subsets of $V$ or {\em blocks} such that the following hold.
\begin{enumerate}[(i)]
    \item For all $b\in\cB$, the size of $b$ is $k$.
    \item Any $t$-subset of $V$ is contained in at most one block in $\cB$.
\end{enumerate}
The {\em size} of a packing refers to the number of blocks in $\cB$.
\end{definition}

We have the following construction of $\orc$-codes from combinatorial packings.

\begin{proposition}[{\cite[Cor.~8.3.3]{du2000combinatorial}}]
\label{prop:or-packing}
    If there exists a $t$-$(v,k,1)$-packing of size $M$, 
    then we have code $\cC$ of length $v$ and size $M$ with $\dist_{\orc}(\cC)\ge k-s(t-1)$.
\end{proposition}

Furthermore, we have the following existential lower bound for $t$-$(v,k,1)$-packings due to Erd\"os, Frankel and F\"uredi.

\begin{theorem}[{see \cite[Thm.~7.3.8]{du2000combinatorial}}]
\label{thm:EFF}
There exists a $t$-$(v,k,1)$-packing of size at least $\binom{n}{t-1}/\binom{k}{t-1}^2$.
\end{theorem}

With these two results, we can now prove Proposition~\ref{prop:lower-or}.

\begin{proof}[Proof of Proposition~\ref{prop:lower-or}]
    It follows from Proposition~\ref{prop:lower-or} and Theorem~\ref{thm:EFF} that
    \[ \maxcode_\orc(n,d) \ge \max \left\{\frac{\binom{n}{t-1}}{\binom{k}{t-1}^2}:~ d\ge k-s(t-1) \right\}\,.\]
    To obtain the asymptotic rates, we set $d=\ceilenv{\delta n}$, $t-1=\ceilenv{\tau n}$ and $k=\ceilenv{\kappa n}$.
    Then we have $\delta = \kappa-s\tau$ or $\tau = \frac 1s(\kappa-\delta)$. 
    Taking logarithms, we obtain Proposition~\ref{prop:lower-or}.
\end{proof}

\section{Enumerating Hyperedges in $\cG(n,d)$}
\label{app:hyperedge}

In this section, we determine the rates $\beta_i$ defined in Corollary~\ref{cor:gvrate}.
To this end, we first define the following sets:
{\small
\begin{align}
\hspace*{-2mm}\cN_2(n,d) & = \{(\bfa,\bfb)\in (\{0,1\}^n)^2 \,:\, \norm{\bfa-\bfb} < d\}, \label{eq:n2}\\
\hspace*{-2mm}\cN_3(n,d) & = \{(\bfa,\bfb,\bfc)\in (\{0,1\}^n)^3 \,:\,\norm{(\bfa \boxplus_\lambda \bfb) -\bfc} < d\}, \label{eq:n3}\\
\hspace*{-2mm}\cN_4(n,d) & = \{(\bfa,\bfb,\bfc,\bfd)\in (\{0,1\}^n)^4 \,: \notag \\ 
&\hspace{2.7cm}
\norm{(\bfa \boxplus_\lambda \bfb)-(\bfc \boxplus_\lambda \bfd)} < d\}, \label{eq:n4}
\end{align}
}
and show that the limiting rates of these sets provide sharp asymptotic estimates of $\beta_i$.

To do so, let $\cD(n,i)$ denote the set of $i$-tuples $(\bfx_1,\bfx_2,\ldots, \bfx_i)$ where $\bfx_j$'s are all pairwise distinct binary words of length $n$.
Then we set $\cN^*_i(n,d) = \cN^*_i(n,d)\cap \cD(n,i)$.

We first show that $N^*_i(n,d)$ can be used to estimate $E_i$.

\begin{lemma}\label{lem:est-1}
For $i\in\{2,3,4\}$, recall that $E_i$ is the number of hyperedges of size $i$ in $\cG(n,d)$. 
Then we have that
\[
\frac1{i!} |\cN^*_i(n,d)| \le E_i \le |\cN^*_i(n,d)|.
\]
\end{lemma}

\begin{proof}
We prove for the case $i=4$ and the proof for the other cases are similar.
For each edge of size four, we can reorder its vertices to form a tuple $(\bfa,\bfb,\bfc,\bfd)\in (\{0,1\}^n)^4$ such that $d(\bfa \boxplus_\lambda \bfb,\bfc \boxplus_\lambda \bfd) <d$. 
Therefore, each edge corresponds to a tuple in $\cN_4^*(n,d)$. 
This yields the upper bound.

On the other hand, for each tuple in $\cN_4^*(n,d)$, since all elements are distinct, we consider them as a 4-subset of $\{0,1\}^n$. 
Furthermore, this subset of vertices forms a hyperedge in $\cE(n,d)$. 
Now, at most $4!=24$ such tuples correspond to the same hyperedge in $\cH_4$, and so, we have the lower bound.
\end{proof}

When $i=2$, we have that $\cN_2(n,d)=\cN^*_2(n,d)\cup \{(\bfa,\bfa):\bfa\in\{0,1\}^n\}$.
So, $|\cN^*_2(n,d)|=|\cN_2(n,d)|-2^n$.
For fixed $\delta$, since $\lim_{n\to\infty}\log_2{|\cN_2(n,d)|}/n = 1+\entropy(\delta)$, 
we apply Lemma~\ref{lem:est-1} to show that $\beta_2=1+\entropy(\delta)$.  

For $i\in\{3,4\}$, determining $N^*_i$ is slightly more tedious.
Nevertheless, we proceed similarly as with $i=2$.
Specifically, we define the following sets:
\begin{align}
\cB_3(n,d) & = \{(\bfa,\bfb,\bfc)\in \cN_3(n,d) \,:\, \bfa=\bfb \}, \label{eq:n31}\\
\cC_3(n,d) & = \{(\bfa,\bfb,\bfc)\in \cN_3(n,d) \,:\, \bfa=\bfc \}, \label{eq:n32}\\
\cB_4(n,d) & = \{(\bfa,\bfb,\bfc,\bfd)\in \cN_4(n,d) \,:\, \bfa=\bfb \}, \label{eq:n41}\\
\cC_4(n,d) & = \{(\bfa,\bfb,\bfc,\bfd)\in \cN_4(n,d) \,:\, \bfa=\bfc \}. \label{eq:n42}
\end{align}

Then we have the following lemmas.

\begin{lemma}\label{lem:est-2}
We have that
{\footnotesize
\begin{align}
|\cN_3(n,d)|- |\cB_3(n,d)|-2|\cC_3(n,d)|&\le |\cN^*_3(n,d)|\le |\cN_3(n,d)|\,,  \label{eq:est-23} \\
|\cN_4(n,d)|-2|\cB_4(n,d)|-4|\cC_4(n,d)|&\le |\cN^*_4(n,d)|\le |\cN_4(n,d)|\,.  \label{eq:est-24}
\end{align}
}
\end{lemma}

\begin{proof}
The inequality $|\cN^*_i(n,d)|\le |\cN_i(n,d)|$ follows from definition.
For the lower bound, as before, we prove for the case $i=4$ and the case for $i=3$ can be proved similarly.
Let $(\bfa,\bfb,\bfc,\bfd)\in\cN_4(n,d)\setminus\cN^*(n,d)$. Then we have either 
\[(\bfa=\bfb \text{ or } \bfc=\bfd )\]
or 
\[(\bfa=\bfc \text{ or } \bfa=\bfd\text{ or } \bfb=\bfc\text{ or } \bfb=\bfd).\]
In the first two sub-cases, we see from symmetry that $|\cB_4(n,d)|=|\{(\bfa,\bfb,\bfc,\bfd)\in \cN_4(n,d) \,:\, \bfa=\bfb \}|$ is equal to 
$|\{(\bfa,\bfb,\bfc,\bfd)\in \cN_4(n,d) \,:\, \bfc=\bfd \}|$.
Similarly, the next four sub-cases, we have these quantities to equal to $|\cC_4(n,d)|$.
Therefore, from union bound, we have that $|\cN_4(n,d)|\le |\cN^*_4(n,d)|+2|\cB_4(n,d)|+4|\cC_4(n,d)|$, as required.
\end{proof}

Therefore, it remains to determine the asymptotic rates of the quantities $|\cN_i(n,d)|$ for $i\in\{3,4\}$. 
To this end, we determine the recursive formulas for our quantities of interest.

Now, we assume $\lambda$ to be rational, and so, we set $\lambda = a/b$ for some integers with $a>b$. 
For $i\in\{3,4\}$, in lieu of $|\cN_i(n,d)|$, we study the quantity 
$N_i(n,d)$ which gives the total number of $4$-tuples $(\bfa,\bfb,\bfc,\bfd)$ with 
$b\norm{(\bfa \boxplus_\lambda \bfb)-(\bfc \boxplus_\lambda \bfd)}=d$.
So, we obtain $|\cN(n,d)|$ through the sum $\sum_{j=0}^{bd-1} N(n,j)$. 
Note that the argument $j$ for $N_i(n,j)$ is always an integer.
We similarly define $B_i(n,d)$ and $C_i(n,d)$.
We are now ready to state the recursive formulas. 

\begin{lemma}
If $n,d\ge 0$ and $(n,d)\ne(0,0)$, then
\begin{align}
N_3(n,d) & = 3N_3(n-1,d) +  N_3(n,d-a) \notag\\ & \hspace*{3mm}+ 3N_3(n,d-b) + N_3(n,d-a+b) ,\\
B_3(n,d) & =  B_3(n-1,d) +  B_3(n,d-a) \notag\\ & \hspace*{3mm}+  B_3(n,d-b) + B_3(n,d-a+b) ,\\
C_3(n,d) & = 2C_3(n-1,d)               \notag\\ & \hspace*{3mm}+  C_3(n,d-b) + C_3(n,d-a+b) ,\\
N_4(n,d) & = 6N_4(n-1,d) + 2N_4(n,d-a) \notag\\ & \hspace*{3mm}+ 4N_4(n,d-b) + 4N_4(n,d-a+b) ,\\
B_4(n,d) & = 2B_4(n-1,d) + 2B_4(n,d-a) \notag\\ & \hspace*{3mm}+ 2B_4(n,d-b) + 2B_4(n,d-a+b) ,\\
C_4(n,d) & = 4C_4(n-1,d) + 2C_4(n,d-a) \notag\\ & \hspace*{3mm}              + 2C_4(n,d-a+b) .
\end{align}
Here, for $i\in\{3,4\}$, the base cases are $N_i(0,0)=B_i(0,0)=C_i(0,0)=1$ and $N_i(n,d)=B_i(n,d)=C_i(n,d)=0$ if $n<0$ or $d<0$.
\end{lemma}

Using these recursive relations, we borrow tools from analytic combinatorics in several variables (ACSV) to determine the corresponding rates. Specifically, we have the following theorem from \cite{pemantle2008twenty}.
Here, we rewrite the result in the form stated in~\cite{goyal2023} (see Eq. 5).

\begin{proposition}
Consider a rational generating function $F(z,u)=\sum_{n,d\ge 0} a_{n,d}z^n u^d=P(z,u)/Q(z,u)$.
Let $\delta>0$. If $a_{n,\delta n}>0$ for all $n$, then there is a unique solution $(z_i^*,u_i^*)>0$ satisfying
\[ Q(z,u) = 0 \text{ and } u\frac{\partial}{\partial u}Q(z,u) =   \delta z\frac{\partial}{\partial z}Q(z,u)\,. \]
Furthermore, 
\[\lim_{n\to\infty} \log a_{n,\delta n}/n = - \log z^* - \delta \log u^*\,.\]
\end{proposition}

Therefore, it suffices to determine the denominators of the corresponding generating functions for $N_i(n,d)$, $B_i(n,d)$ and $C_i(n,d)$. Using standard methods in combinatorics (see for example~\cite{pemantle2008twenty}), we see that the denominators corresponding to $N_i(n,d)$, $B_i(n,d)$, $C_i(n,d)$ are the bivariate polynomials $H_i(z,u)$, $G_{(i,1)}(z,u)$, and $G_{(i,2)}(z,u)$, respectively defined in Theorem~\ref{thm:lamb-GV}.

\end{document}